%% file: arxiv.tex
\documentclass[11pt]{article}

\usepackage[letterpaper,margin=1in]{geometry}
\usepackage[T1]{fontenc}
\usepackage[utf8]{inputenc}
\usepackage{amsmath,amssymb,amsthm,mathtools}
\usepackage{aliascnt}
\usepackage{microtype}
\usepackage{enumitem}
\usepackage{mleftright}
\usepackage{booktabs}
\usepackage{xcolor}
\usepackage[authoryear]{natbib}
\usepackage[colorlinks=true,linkcolor=magenta,citecolor=teal,urlcolor=magenta]{hyperref}
\usepackage{algorithm}
\usepackage{algpseudocode}

\setlist[itemize]{leftmargin=1.6em,itemsep=0.25em,topsep=0.25em}
\setlist[enumerate]{leftmargin=1.8em,itemsep=0.25em,topsep=0.25em}
\allowdisplaybreaks
\usepackage[capitalize,noabbrev]{cleveref}

\usepackage{xcolor}
\definecolor{myOrange}{RGB}{235, 105, 92}
\definecolor{myMaroon}{RGB}{128, 0, 0}
\definecolor{maroon}{rgb}{0.5,0.0,0.0}
\definecolor{heavymaroon}{rgb}{0.3,0.0,0.0}

\definecolor{mygreen}{RGB}{160, 200, 140}

\usepackage{hyperref}
\hypersetup{ 
colorlinks=true, 
linkcolor=myOrange, 
citecolor=myMaroon, filecolor=cyan, 
urlcolor=mygreen
}

\makeatletter
\let\oldalglinenumber\alglinenumber
\renewcommand{\alglinenumber}[1]{%
  \protected@edef\@currentlabel{#1}%
  \oldalglinenumber{#1}%
}
\makeatother

\input{notation}

\title{Breaking the 1/3 Barrier for Approximate Nash Equilibria in Symmetric Bimatrix Games}
\author{
Ashkan Soleymani\\
MIT\\
\texttt{ashkanso@mit.edu}
\and
Patrick Jaillet\\
MIT\\
\texttt{jaillet@mit.edu}
\and
Gabriele Farina\\
MIT\\
\texttt{gfarina@mit.edu}
}

\date{}

\begin{document}
\maketitle

\begin{abstract}
We give the first unconditional family of polynomial-time algorithms
that breaks the $1/3$ approximation barrier for symmetric bimatrix
games. For every fixed $\delta>0$, an algorithm in this family
computes a symmetric $(1-1/\sqrt{2}+\delta)$-approximate Nash
equilibrium, improving on the previous best guarantee of
$(1/3+\delta)$ due to \citet*{kontogiannis2011approximability}.
\end{abstract}

\tableofcontents
\newpage

\section{Introduction} \label{sec:intro}

Nash equilibrium is the standard notion of stability in strategic
interaction. Nash's theorem guarantees that every finite game
admits a mixed-strategy equilibrium~\citep*{nash1950equilibrium}.
However, finding an equilibrium is $\mathsf{PPAD}$-complete even in
two-player bimatrix games%
~\citep*{daskalakis2009complexity,chen2009settling}.

The hardness of computing exact equilibria has motivated the study
of additive approximations. With payoffs normalized to $[0,1]$,
an $\eps$-approximate Nash equilibrium is a mixed-strategy profile
in which neither player can gain more than $\eps$ by deviating
unilaterally. The small-support theorem of
\citet*{lipton2003playing} yields a quasi-polynomial-time
approximation scheme for bimatrix games.
\citet{rubinstein2017settling} showed that this running time is
nearly optimal under the exponential-time hypothesis for
$\mathsf{PPAD}$. Specifically, under this assumption, there exists
a constant $\eps>0$ for which computing an $\eps$-approximate
Nash equilibrium of an $n\times n$ bimatrix game requires time
$n^{\log^{1-o(1)} n}$. The same assumption also rules out a PTAS
for general bimatrix games. This motivates the search for the
smallest approximation error achievable in polynomial time.

The class of \emph{symmetric bimatrix games} provides a natural
setting for this question. We write $\mathcal G\coloneqq(A,A^\top)$ for a
game in this class, where $A\in[0,1]^{n\times n}$ is the row
player's payoff matrix and $A^\top$ is the column player's
payoff matrix. At a pure strategy profile $(i,j)$, the row player
receives $A_{ij}$ and the column player receives $A_{ji}$.
Every symmetric bimatrix game admits a symmetric Nash equilibrium,
in which both players use the same mixed
strategy~\citep*{nash1951non}. However, symmetry does not remove
the hardness of exact equilibrium computation. The classical
symmetrization reduction of \citet*{gale1950symmetric}, together
with the $\mathsf{PPAD}$-completeness of bimatrix
Nash~\citep*{daskalakis2009complexity,chen2009settling},
implies that computing a symmetric Nash equilibrium is also
$\mathsf{PPAD}$-complete. This motivates studying whether
symmetry allows better polynomial-time approximation guarantees.

Symmetry nevertheless simplifies the approximation problem. A single mixed
strategy $p$ defines the symmetric profile $(p,p)$, in which both
players receive the same payoff and have the same regret, given by
\[
    \operatorname{Reg}(p)
    \coloneqq \max_{i\in[n]} e_i^\top A p - p^\top A p.
\]
Thus, finding a symmetric $\eps$-approximate Nash equilibrium
amounts to finding a strategy $p$ whose best-response payoff $\max_{i\in[n]} e_i^\top A p$
exceeds its self-payoff $p^\top A p$ by at most $\eps$, or equivalently,
$\operatorname{Reg}(p)\leq\eps$. Despite this simplification,
achieving an approximation guarantee below $1/3$ in polynomial
time had remained elusive.

\paragraph{The one-third frontier.}
For every fixed $\delta>0$, \citet*{kontogiannis2011approximability}
gave a polynomial-time algorithm that computes a symmetric
$(1/3+\delta)$-approximate Nash equilibrium. Their quadratic-programming
formulation also yields an approximate equilibrium with regret strictly
below $1/3$ given an exact KKT point. However, their polynomial-time
algorithm computes only an approximate KKT point and guarantees
$1/3+\delta$. \citet{ahmadi2021semidefinite} later showed that a rank-two solution
of an SDP relaxation can be rounded to a $1/3$-approximate Nash
equilibrium for a symmetric game.
This guarantee is conditional, since the relaxation is not guaranteed
to return a rank-two solution. Thus, prior to our work, the best
known unconditional polynomial-time guarantee for symmetric
bimatrix games remained $1/3+\delta$.

In parallel, a line of work on general bimatrix games steadily
improved polynomial-time approximation guarantees. Early work by
\citet*{daskalakis2006note} established a $1/2$ guarantee,
followed by the $(3-\sqrt{5})/2\approx0.38197$ guarantee of
\citet*{daskalakis2007progress} and the
$1/2-1/(3\sqrt{6})\approx0.36392$ guarantee of
\citet*{bosse2010new}. \citet*{tsaknakis2008optimization}
developed an optimization-based approach achieving
$0.3393+\delta$, which \citet*{deligkas2023polynomial}
later refined to $1/3+\delta$ for every bimatrix game.
Symmetric games had reached the $1/3+\delta$ guarantee
more than a decade earlier, but no unconditional
polynomial-time algorithm was known to break the $1/3$
barrier for either class. Thus, prior to our work,
$1/3+\delta$ remained the best known unconditional
polynomial-time guarantee for both general and symmetric
bimatrix games.

\paragraph{Our contribution.}
We give the first unconditional polynomial-time algorithm that
breaks the $1/3$ approximation barrier for symmetric bimatrix
games. One concrete instantiation computes a symmetric
$3/10$-approximate Nash equilibrium. More generally, for every
fixed $\delta>0$, an algorithm in our family computes a symmetric
$\eps$-approximate Nash equilibrium with
\[
    \eps = 1-\frac{1}{\sqrt{2}}+\delta,
\]
where $1-1/\sqrt{2}\approx0.29289<1/3$.
The algorithm returns a single mixed strategy $p$, yielding the
symmetric profile $(p,p)$. Thus, our result improves the best
known unconditional polynomial-time approximation guarantee for
symmetric bimatrix games from $1/3+\delta$ to
$1-1/\sqrt{2}+\delta$.

Unlike previous approaches based on quadratic programming or descent
methods, our algorithm uses auxiliary zero-sum games as its only
optimization subroutines. We begin with the \emph{saddle-gap game}
$\mathcal S$, an auxiliary zero-sum game in which the row player
chooses an ordered pair $(i,j)$, the column player chooses a pure
strategy $k$, and the payoff is $A_{ik}-A_{kj}$. If the value $v$
of $\mathcal S$ is at most $\eps$, an optimal column strategy
yields a symmetric $\eps$-approximate Nash equilibrium. Otherwise,
the first marginal $a$ of an optimal row strategy of $\mathcal S$
has self-payoff at least $2v$ and earns at least $v$ against any
opponent strategy. The self-payoff bound alone recovers the
one-third benchmark of \citet*{kontogiannis2011approximability}.
The second bound allows us to control the self-payoff when
modifying $a$, providing the additional information needed
to improve on this guarantee.

Our main technical contribution is an amortized repair procedure
built around a sequence of further auxiliary zero-sum games.
At iteration $t$, the algorithm identifies the near-best-response
set $B_t$ of the current strategy $p_t$ and solves the restricted
zero-sum game $\mathcal B_t$ on $B_t\times B_t$, with payoffs
inherited from $A$. If its value is at least $1-\eps$, an optimal
row strategy directly yields a symmetric $\eps$-approximate Nash
equilibrium. Otherwise, an optimal column strategy provides a
repair direction that controls the payoffs of the near-best
responses. The remaining strategies have enough slack to stay
controlled after the update, while the guarantee inherited from
$\mathcal S$ helps preserve the self-payoff. Our amortized analysis
shows how to balance these effects over successive repairs.
A single repair already achieves $\sqrt{7}/2-1\approx0.32288$,
breaking the $1/3$ barrier. Repeated repairs yield the concrete
$3/10$ guarantee and approach $1-1/\sqrt{2}$ in the limit.

\subsection{Our Results} \label{sec:main}

Our main theorem gives a two-parameter family of algorithms. The
parameter $K$ bounds the number of repair steps, while $m$ determines
the repair weights $\alpha_t\coloneqq1/(m+t)$.

\begin{theorem}[Main result]\label{thm:finite-main}
For every pair of integers $K\ge1$ and $m\ge2$, there is a
polynomial-time algorithm that, given a symmetric bimatrix game
$\mathcal G\coloneqq(A,A^\top)$ with $A\in[0,1]^{n\times n}$, computes
a symmetric $\eps_{K,m}$-approximate Nash equilibrium, where
\[
\numberthis{eq:epsilon-K-m}
\eps_{K,m}\coloneqq
\frac{2m^2+2(K-2)m+K(K-1)+2}
{6m^2+(8K-12)m+3(K-1)(K-2)} .
\]
The algorithm solves the auxiliary zero-sum game $\mathcal S$,
with $n^2$ row actions and $n$ column actions, followed by at
most $K$ restricted zero-sum games $\mathcal B_t$, each with
at most $n$ actions per player.
\end{theorem}

The theorem gives a concrete improvement over the $1/3$ guarantee
with only a finite number of repairs. In particular, choosing
$K=13$ and $m=14$ yields $\eps_{13,14}=3/10$, using at most
$13$ repair steps.

\begin{corollary}\label{cor:three-tenths}
Every symmetric bimatrix game with payoffs in $[0,1]$ has a
polynomial-time computable symmetric $3/10$-approximate Nash equilibrium.
\end{corollary}

The guarantee improves further as the number of repairs increases.
Taking $K\to\infty$ with $m/K\to1/\sqrt{2}$ yields
$\eps_{K,m}\to1-1/\sqrt{2}\approx0.29289$.
Since $K$ and $m$ can be chosen as constants depending only
on $\delta$, this gives the following polynomial-time guarantee.

\begin{corollary}\label{cor:limit}
For every fixed $\delta>0$, there is a polynomial-time algorithm
that computes a symmetric
$\left(1-\frac1{\sqrt2}+\delta\right)$-approximate Nash equilibrium
of every symmetric bimatrix game with payoffs in $[0,1]$.
\end{corollary}

The improvement begins with a single repair, which achieves
$\sqrt{7}/2-1\approx0.32288$ with a suitable choice of weight.
The multi-repair analysis builds on this idea by preserving and
exploiting the payoff guarantee from $\mathcal S$ across successive
repairs, leading to the stronger bounds in \Cref{thm:finite-main}.

\subsection{Technical Overview and Proof Strategy}
\label{sec:intro-technical-overview}

The algorithm maintains a single mixed strategy that is updated
through successive repairs. Throughout this overview, we denote
the current strategy by $p$, corresponding to $p_t$ at iteration
$t$ in the formal algorithm.

\paragraph{Saddle-gap game.}
We begin with the auxiliary zero-sum game $\mathcal S$, which we call
the \emph{saddle-gap game}. Its row player chooses an ordered pair
$(i,j)$, its column player chooses a pure strategy $k$, and the row
player's payoff is $A_{ik}-A_{kj}$. For a mixed column strategy $x$ of $\mathcal S$, the row player's
best-response payoff is
\[
    \max_i e_i^\top A x-\min_j x^\top A e_j.
\]
This quantity upper-bounds the regret of the symmetric profile
$(x,x)$ in the original game $\mathcal G$, since
$x^\top A x\ge\min_j x^\top A e_j$. In particular, if the
value $v$ of $\mathcal S$ is at most $\eps$, an optimal column
strategy $x$ of $\mathcal S$ satisfies
\[
    \operatorname{Reg}(x)
    =\max_i e_i^\top A x-x^\top A x
    \le \max_i e_i^\top A x-\min_j x^\top A e_j
    =v\le\eps.
\]
Thus, $(x,x)$ is a symmetric $\eps$-approximate Nash equilibrium
of $\mathcal G$.

Suppose now that $v>\eps$. Let $\sigma$ be an optimal row strategy
of $\mathcal S$, with first and second marginals $a$ and $b$.
By the maximin property of $\sigma$, for every $k\in[n]$,
\[
    a^\top A e_k-e_k^\top A b\ge v.
\]
Since all payoffs are nonnegative, $a^\top A e_k\ge v$ for every $k$.
Averaging the same inequality over $k\sim a$ and $k\sim b$ gives
\[
    a^\top A a\ge a^\top A b+v
    \ge b^\top A b+2v\ge 2v.
\]
We initialize the repair procedure at $p_0\coloneqq a$. Since all payoffs are
at most one, these bounds give
\[
    \operatorname{Reg}(p_0)\le 1-2v<1-2\eps,
    \qquad
    p_0^\top A e_j\ge v>\eps\quad\forall j\in[n].
\]
The regret bound alone already recovers the one-third guarantee of
\citet*{kontogiannis2011approximability}, since
$1-2\eps=\eps$ when $\eps=1/3$. To improve on this guarantee,
we exploit the additional property
$p_0^\top A e_j\ge v>\eps$ for every $j\in[n]$.
This lower bound also holds against any mixed column strategy,
allowing us to control the self-payoff during repairs.
Our analysis tracks how much of this guarantee is preserved
across successive updates.

\paragraph{A generic repair step.}
Consider a current strategy $p$ and a repair weight $\alpha\in(0,1)$.
The algorithm seeks a strategy $z$ to mix with $p$, assigning it
weight $\alpha$. To identify the rows that need to be controlled
directly, we define the near-best-response set
\[
\numberthis{eq:intro-near-best}
B\coloneqq\left\{i\in[n]:
e_i^\top A p\ge
\max_h e_h^\top A p-
\frac{\alpha\eps}{1-\alpha}
\right\}.
\]
The set $B$ contains every best response to $p$. Rows outside $B$
have a payoff gap against $p$ that will help control their payoffs
after the update.

We solve the auxiliary zero-sum game $\mathcal B$ with row and
column action sets $B$ and payoff matrix $A_{B,B}$. If its value
is at least $1-\eps$, minimax gives a strategy $w\in\Delta(B)$
such that
\[
    w^\top A e_j\ge1-\eps\qquad\forall j\in B.
\]
Since $w$ is supported on $B$, we have $w^\top A w\ge1-\eps$.
Every pure deviation earns at most one, so
$\operatorname{Reg}(w)\le\eps$. The algorithm therefore stops
and outputs the symmetric $\eps$-approximate Nash equilibrium
$(w,w)$ of $\mathcal G$. 

Otherwise, an optimal column strategy $z\in\Delta(B)$ of
$\mathcal B$ satisfies
\[
    e_i^\top A z\le1-\eps\qquad\forall i\in B.
\]
We use $z$ as a repair direction and update $p$ to
\[
\numberthis{eq:intro-p-plus}
p^+\coloneqq(1-\alpha)p+\alpha z.
\]
The choice of the near-best-response window in \eqref{eq:intro-near-best} ensures that every
pure row has a controlled payoff against $p^+$. If $i\in B$,
the minimax guarantee for $z$ gives
\[
\numberthis{eq:reg_i_in}
e_i^\top A p^+
\le (1-\alpha)\max_h e_h^\top A p+\alpha(1-\eps).
\]
If $i\notin B$, the definition of $B$ and the upper bound
$e_i^\top A z\le1$ instead give
\[
\numberthis{eq:reg_i_out}
e_i^\top A p^+
\le (1-\alpha)\left(\max_h e_h^\top A p-
\frac{\alpha\eps}{1-\alpha}\right)+\alpha.
\]
The right-hand sides of \eqref{eq:reg_i_in} and
\eqref{eq:reg_i_out} coincide. This explains the width
$\alpha\eps/(1-\alpha)$ in \eqref{eq:intro-near-best} and yields
\[
\max_{i\in[n]}e_i^\top A p^+
\le (1-\alpha)\max_h e_h^\top A p+\alpha(1-\eps).
\]

\paragraph{Tracking regret through repairs.}
The preceding argument bounds the best-response payoff against $p^+$.
To bound the regret of $(p^+,p^+)$ in $\mathcal G$, we also need
to control its self-payoff. For this purpose, we track two bounds
for the current strategy $p$, namely a regret budget $R$ and a
lower bound $L$ on its payoff against every pure column strategy,
\[
    \operatorname{Reg}(p)
    =\max_i e_i^\top A p-p^\top A p\le R,
    \qquad
    p^\top A e_j\ge L\quad\forall j\in[n].
\]
Using the update in \eqref{eq:intro-p-plus}, we expand the
self-payoff of the repaired strategy as
\[
    (p^+)^\top A p^+
    =(1-\alpha)^2p^\top A p +\alpha(1-\alpha)(p^\top A z+z^\top A p)
    +\alpha^2z^\top A z.
\]
Both cross terms in $p^\top A z+z^\top A p$ admit useful lower bounds. Since $z$ is
supported on the near-best-response set $B$, we have
\[
    z^\top A p
    &\ge \max_h e_h^\top A p-
    \frac{\alpha\eps}{1-\alpha}.
\]
Moreover, since $p^\top A e_j\ge L$ for every pure column
strategy $j$, we also have
\[
    p^\top A z\ge L.
\]
Combining these bounds with the best-response estimate from
the preceding paragraph and using $z^\top A z\ge0$, we obtain
\[
\numberthis{eq:intro-one-step}
    \operatorname{Reg}(p^+)
    &=\max_i e_i^\top A p^+-(p^+)^\top A p^+\\
    &\le (1-\alpha)^2R+\alpha(1-\eps)+\alpha^2\eps
    -\alpha(1-\alpha)L.
\]
The contribution of the previous regret is scaled by
$(1-\alpha)^2$, while the term $-\alpha(1-\alpha)L$ reflects
the lower bound on $p^\top A z$.

The same repair preserves part of this lower bound for subsequent
iterations. Indeed, since all payoffs are nonnegative,
\[
    (p^+)^\top A e_j
    =(1-\alpha)p^\top A e_j+\alpha z^\top A e_j \ge(1-\alpha)L\qquad\forall j\in[n].
\]
Thus, each repair provides a new regret bound while preserving
a fraction of the payoff guarantee needed for the next step.
This allows the procedure to be repeated, with the remaining
bound contributing to the analysis of subsequent repairs.

\paragraph{From one repair to amortized repairs.}
As a warm-up, we analyze the single-repair algorithm in
\Cref{sec:warmup} by applying \eqref{eq:intro-one-step} once.
Starting from the saddle-gap marginal $a$, we use the initial
budgets $R\coloneqq1-2\eps$ and $L\coloneqq\eps$. For the repaired strategy
$p^+ = (1-\alpha)a+\alpha z$, the one-step inequality gives
\[
    \operatorname{Reg}(p^+)
    &\le 1-\alpha+\alpha^2-2\eps(1-\alpha).
\]
To ensure that the regret is at most $\eps$, it suffices to choose
\[
    \eps\ge\frac{1-\alpha+\alpha^2}{3-2\alpha}.
\]
Minimizing the right-hand side over $\alpha\in(0,1)$ shows
that a single repair suffices to compute a symmetric
$\eps$-approximate Nash equilibrium with
$\eps=\sqrt7/2-1\approx0.32288$, already below the
$1/3$ benchmark.

We now extend the single-repair procedure to multiple repairs.
The key observation is that the payoff lower bound against every
pure column strategy is partially preserved after each update,
allowing us to use it again in subsequent repairs. Whenever
iteration $t$ reaches the repair branch, the restricted zero-sum
game $\mathcal B_t$ provides a repair direction $z_t$. We choose
the weights $\alpha_t\coloneqq1/(m+t)$ and update
\[
    p_{t+1}
    &\coloneqq\left(1-\frac1{m+t}\right)p_t+\frac1{m+t}z_t.
\]
To analyze these updates, let $R_t$ and $L_t$ denote the regret
and column-floor budgets for $p_t$, respectively. Applying the
one-step bounds at each repair gives
\[
    R_{t+1}
    &\coloneqq\left(1-\frac1{m+t}\right)^2R_t
      +\frac{1-\eps}{m+t}
      +\frac{\eps}{(m+t)^2} -\frac1{m+t}\left(1-\frac1{m+t}\right)L_t,\\
    L_{t+1}
    &\coloneqq\left(1-\frac1{m+t}\right)L_t,
\]
with initial values $R_0\coloneqq1-2\eps$ and $L_0\coloneqq\eps$, as guaranteed
by the saddle-gap certificate.

The chosen repair weights make both recurrences telescope.
The floor recurrence gives
\[
    L_t=\eps\frac{m-1}{m+t-1}.
\]
Substituting this expression into the regret recurrence and
multiplying by $(m+t)^2$ yields
\[
    (m+t)^2R_{t+1}
    &=(m+t-1)^2R_t
      +(m+t)-\eps(2m+t-2).
\]
Summing over $t=0,1,\ldots,K-1$ cancels the intermediate
scaled regret terms, leaving only scaled versions of the initial regret $(m-1)^2R_0$ and
final regret $(m+K-1)^2R_K$. Choosing $\eps$ so that the final regret
budget satisfies $R_K=\eps$ yields exactly
$\eps=\eps_{K,m}$ in \eqref{eq:epsilon-K-m}.

Consequently, after at most $K$ repairs, the algorithm
computes a symmetric $\eps_{K,m}$-approximate Nash equilibrium
of $\mathcal G$. In particular, choosing $K=13$ and $m=14$
gives $\eps_{13,14}=3/10$, yielding a symmetric
$3/10$-approximate Nash equilibrium. Moreover, taking
$K\to\infty$ with $m/K\to1/\sqrt2$ gives
$\eps_{K,m}\to1-1/\sqrt2$. Thus, for every fixed $\delta>0$,
we obtain a polynomial-time algorithm that computes a symmetric
$(1-1/\sqrt2+\delta)$-approximate Nash equilibrium.

\section{Related Work}\label{sec:related_work}

Computing an exact Nash equilibrium is $\mathsf{PPAD}$-complete
~\citep*{daskalakis2009complexity,chen2009settling}. The
symmetrization reduction of \citet*{gale1950symmetric} extends
this hardness to computing symmetric equilibria in symmetric
bimatrix games. For approximate equilibria,
\citet*{lipton2003playing} established the existence of
$\eps$-approximate equilibria with logarithmic support, yielding
an $n^{O(\log n/\eps^2)}$-time algorithm.
\citet*{rubinstein2017settling} showed that, under the
exponential-time hypothesis for $\mathsf{PPAD}$, computing an
$\eps$-approximate Nash equilibrium for some constant $\eps>0$
requires time $n^{\log^{1-o(1)} n}$. These results motivate
the study of approximation guarantees achievable in polynomial
time.

For symmetric bimatrix games, \citet*{kontogiannis2011approximability}
developed a quadratic-programming formulation whose exact KKT
points yield approximate equilibria with regret strictly below
$1/3$. Given such a point, the corresponding approximate
equilibrium can be recovered efficiently. Their polynomial-time
algorithm computes an approximate KKT point and guarantees a
symmetric $(1/3+\delta)$-approximate Nash equilibrium for every
fixed $\delta>0$. Thus, the sub-third guarantee requires access
to an exact KKT point, while the unconditional polynomial-time
guarantee remains $1/3+\delta$.

\citet*{ahmadi2021semidefinite} subsequently studied an SDP
relaxation of a quadratic formulation for Nash equilibria in
bimatrix games. A rank-one solution yields an exact Nash
equilibrium, while a rank-two solution can be rounded to a
$1/3$-approximate Nash equilibrium for a symmetric game.
Since the relaxation is not guaranteed to return a solution
of rank at most two, this guarantee is conditional.
The algorithm of \citet*{deligkas2023polynomial} achieves a
$(1/3+\delta)$-approximate Nash equilibrium for every fixed
$\delta>0$ in general bimatrix games and therefore also applies
to symmetric games. However, even when the input game is symmetric, their algorithm
does not require both players to use the same mixed strategy.
Consequently, the resulting approximate Nash equilibrium is
not guaranteed to be symmetric. To the best of our knowledge,
our algorithm is the first unconditional polynomial-time
algorithm to compute a symmetric $\eps$-approximate Nash
equilibrium with $\eps<1/3$ for every symmetric bimatrix game.

The polynomial-time approximation frontier for arbitrary bimatrix
games has advanced through a long sequence of results.
\citet*{kontogiannis2006polynomial} gave an early $3/4$-approximation.
\citet*{daskalakis2006note} improved this to $1/2$, followed
by the $(3-\sqrt{5})/2\approx0.38197$ guarantee of
\citet*{daskalakis2007progress}.
\citet*{bosse2010new} improved the guarantee to
$1/2-1/(3\sqrt{6})\approx0.36392$.
Using a descent method for a nonconvex regret objective,
\citet*{tsaknakis2008optimization} achieved $0.3393+\delta$,
which remained the best known guarantee for roughly fifteen
years. \citet*{deligkas2023polynomial} later refined this
approach to obtain a $(1/3+\delta)$-approximate Nash equilibrium
for every fixed $\delta>0$. Our result improves on this
guarantee for symmetric bimatrix games while also ensuring
that the output is a symmetric approximate Nash equilibrium.

An $\eps$-well-supported Nash equilibrium ($\eps$-WSNE) requires
every pure strategy played with positive probability by either
player to be within $\eps$ of a best response. This is stronger
than an ordinary $\eps$-approximate Nash equilibrium, which only
bounds each player's gain from a unilateral deviation.
For general bimatrix games, \citet*{kontogiannis2010well} gave
a polynomial-time $2/3$-WSNE algorithm.
\citet*{fearnley2016approximate} improved the guarantee to
approximately $0.6608$, and \citet*{czumaj2019distributed}
obtained $0.6528$. \citet*{deligkas2023polynomialwell}
later achieved a $1/2$-WSNE guarantee in polynomial time
for every bimatrix game.

For symmetric bimatrix games, \citet*{czumaj2014approximate}
gave a polynomial-time $(1/2+\delta)$-WSNE algorithm.
In subsequent work, \citet*{czumaj2017zero} developed an
algorithm for arbitrary bimatrix games that outputs either
a $1/3$-approximate Nash equilibrium or a $1/2$-WSNE.
Since a $1/2$-WSNE only guarantees ordinary regret at most
$1/2$, this does not provide a uniform $1/3$-approximation
guarantee for ordinary Nash equilibrium.

These WSNE results are complementary to our work, as they
address a stronger notion of equilibrium. Their guarantees
do not improve the $1/3$ benchmark for ordinary approximate
Nash equilibria. Our algorithm breaks this benchmark for
symmetric bimatrix games while returning a symmetric profile.

\section{Preliminaries}\label{sec:preliminaries}

Let $[n]\coloneqq\{1,\ldots,n\}$, and let $\Delta_n$ denote the
simplex of probability distributions over $[n]$. For any nonempty
set $S\subseteq[n]$, we write $\Delta(S)$ for the distributions
supported on $S$. We denote by $e_i$ the $i$th standard basis
vector, identified with pure strategy $i$. Throughout, we assume
that the input payoff matrices have entries in $[0,1]$.

A bimatrix game with $n$ strategies per player is specified by
a pair $(R,C)$ of payoff matrices. At a mixed-strategy profile
$(x,y)\in\Delta_n\times\Delta_n$, the row and column players
receive payoffs $x^\top R y$ and $x^\top C y$, respectively.
Their corresponding regrets are
\[
    \max_{i\in[n]}e_i^\top R y-x^\top R y,
    \qquad \textrm{and} \qquad 
    \max_{j\in[n]}x^\top C e_j-x^\top C y.
\]
A profile is an $\eps$-approximate Nash equilibrium if both
players have regret at most $\eps$.

For the symmetric input game $\mathcal G=(A,A^\top)$,
both players receive the same payoff $p^\top A p$ at a
symmetric profile $(p,p)$. Moreover, the column player's
payoff from deviating to pure strategy $j$ is
\[
    p^\top A^\top e_j=e_j^\top A p,
\]
which equals the row player's payoff from the same deviation.
Thus, both players have the same regret at any symmetric
profile, as formalized in the following lemma.

\begin{lemma}\label{lem:symmetric-regret}
For every $p\in\Delta_n$, the symmetric profile $(p,p)$ is an
$\eps$-approximate Nash equilibrium of $\mathcal G$ if and
only if
\[
    \operatorname{Reg}(p)
    = \max_{i\in[n]}e_i^\top A p-p^\top A p\le\eps.
\]
\end{lemma}

\begin{proof}
At $(p,p)$, both players receive payoff $p^\top A p$ and
have the same best-response payoff $\max_{i\in[n]}e_i^\top A p$.
Their common regret is therefore $\operatorname{Reg}(p)$,
and the claim follows.
\end{proof}

We also use the finite minimax theorem in the following form
for the auxiliary zero-sum games considered in our algorithms.

\begin{lemma}\label{lem:minimax-operational}
Let $M$ be the row player's payoff matrix in a finite
zero-sum game, and let $T\in\mathbb{R}$. If the value of $M$
is at least $T$, then the row player has a mixed strategy
that guarantees a payoff of at least $T$ against every pure
column strategy. If the value of $M$ is less than $T$,
then the column player has a mixed strategy against which
every pure row strategy receives a payoff of at most $T$.
\end{lemma}

\begin{proof}
The claim follows from the finite minimax theorem~\citep{vonNeumann1928}. Optimal
strategies for both players can be computed by linear programming.
\end{proof}

\section{The Saddle-Gap Game}\label{sec:saddle}

Given the symmetric game $\mathcal G=(A,A^\top)$, we define
an auxiliary zero-sum game $\mathcal S$, called the
\emph{saddle-gap game}. Its row player chooses an ordered pair
$(i,j)\in[n]^2$, its column player chooses a pure strategy
$k\in[n]$, and its payoff matrix $H\in\mathbb R^{n^2\times n}$
is defined by
\[
    H_{(i,j),k}\coloneqq A_{ik}-A_{kj}.
\]
For a mixed column strategy $x\in\Delta_n$, the row player's
best-response payoff in $\mathcal S$ is
\[
    \max_{i,j}\sum_k x_k(A_{ik}-A_{kj})
    &=\max_{i,j}\left(e_i^\top A x-x^\top A e_j\right)\notag\\
    &=\max_i e_i^\top A x-\min_j x^\top A e_j.
    \numberthis{eq:saddle-gap-expression}
\]
This quantity upper-bounds the regret of the symmetric profile
$(x,x)$ in $\mathcal G$, since
$x^\top A x\ge\min_j x^\top A e_j$.

Let $v$ denote the value of $\mathcal S$. By the minimax theorem,
\[
    v
    &\coloneqq
    \min_{x\in\Delta_n}
    \max_{(i,j)\in[n]^2}
    \sum_k x_k(A_{ik}-A_{kj})
    =
    \max_{\sigma\in\Delta_{n^2}}
    \min_{k\in[n]}
    \sum_{i,j}\sigma_{ij}(A_{ik}-A_{kj}).
\]
Let $x\in\Delta_n$ be an optimal column strategy and
$\sigma\in\Delta_{n^2}$ an optimal row strategy of $\mathcal S$.
We denote the first and second marginals of $\sigma$ by
$a,b\in\Delta_n$, respectively, where
\[
    a_i\coloneqq\sum_j\sigma_{ij},
    \qquad
    b_j\coloneqq\sum_i\sigma_{ij}.
\]

The value $v$ determines how the algorithm proceeds.
If $v\le\eps$, the optimal column strategy $x$ already yields
an approximate equilibrium. Otherwise, the row strategy
$\sigma$ provides the initial guarantees needed for the
repair procedure.

\begin{lemma}\label{lem:low-saddle-gap}
If $v\le\eps$, then $(x,x)$ is a symmetric
$\eps$-approximate Nash equilibrium of $\mathcal G$.
\end{lemma}

\begin{proof}
By the optimality of $x$ and \eqref{eq:saddle-gap-expression},
\[
    \max_i e_i^\top A x-\min_j x^\top A e_j
    =v\le\eps.
\]
Since $x^\top A x\ge\min_j x^\top A e_j$, we have
\[
    \operatorname{Reg}(x)
    =\max_i e_i^\top A x-x^\top A x
    \le\max_i e_i^\top A x-\min_j x^\top A e_j 
    &\le\eps.
\]
The claim follows from \Cref{lem:symmetric-regret}.
\end{proof}

\begin{lemma}\label{lem:high-saddle-gap}
For every pure strategy $k\in[n]$,
\[
\numberthis{eq:dual-cert}
    a^\top A e_k-e_k^\top A b\ge v.
\]
Consequently,
\[
\numberthis{eq:column-floor-v}
    a^\top A e_k\ge v\qquad\forall k\in[n],
\]
and
\[
\numberthis{eq:self-payoff-v}
    a^\top A a\ge2v.
\]
\end{lemma}

\begin{proof}
Since $\sigma$ is an optimal row strategy of $\mathcal S$,
it guarantees a payoff of at least $v$ against every pure
column strategy $k\in[n]$. Hence,
\[
    \sum_{i,j}\sigma_{ij}(A_{ik}-A_{kj})\ge v
    \qquad\forall k\in[n].
\]
Expressing this inequality in terms of the marginals $a$
and $b$ gives \eqref{eq:dual-cert}. Since the entries of
$A$ are nonnegative, $e_k^\top A b\ge0$, which proves
\eqref{eq:column-floor-v}.

To establish the self-payoff bound, we average
\eqref{eq:dual-cert} over $k\sim a$ and $k\sim b$,
respectively, obtaining
\[
    a^\top A a-a^\top A b \ge v, \quad \textrm{ and } \quad 
    a^\top A b-b^\top A b \ge v.
\]
Adding these inequalities yields
\[
    a^\top A a-b^\top A b\ge2v.
\]
Since $b^\top A b\ge0$, we obtain
$a^\top A a\ge2v$, proving \eqref{eq:self-payoff-v}.
\end{proof}

When $v>\eps$, the bounds in \Cref{lem:high-saddle-gap}
give the initial regret and payoff guarantees used by
the repair algorithms.

\begin{corollary}\label{cor:initial-certificate}
If $v>\eps$, then $p_0\coloneqq a$ satisfies
\[
    \max_i e_i^\top A p_0-p_0^\top A p_0\le1-2\eps,
    \qquad
    p_0^\top A e_j\ge\eps\quad\forall j\in[n].
\]
\end{corollary}

\begin{proof}
The second inequality follows from \eqref{eq:column-floor-v}
and $v>\eps$. Moreover, \eqref{eq:self-payoff-v} gives
\[
    p_0^\top A p_0=a^\top A a\ge2v>2\eps.
\]
Since every pure best-response payoff is at most one,
we have
\[
    \max_i e_i^\top A p_0-p_0^\top A p_0\le1-2\eps,
\]
as required.
\end{proof}

\begin{algorithm}[!ht]
\caption{Saddle-Gap Algorithm with Single Repair}
\label{alg:warmup-one-repair}
\begin{algorithmic}[1]
\Require A symmetric game $\mathcal G=(A,A^\top)$ with
$A\in[0,1]^{n\times n}$ and a repair weight $\alpha\in(0,1)$

\State Set $\eps\coloneqq\eps(\alpha)$.

\State Solve the saddle-gap game $\mathcal S$ with payoff matrix $H$
and value
\[
    v
    =\min_{x\in\Delta_n}
      \max_{(i,j)\in[n]^2}
      \sum_k x_k(A_{ik}-A_{kj})
    =\max_{\sigma\in\Delta_{n^2}}
      \min_{k\in[n]}
      \sum_{i,j}\sigma_{ij}(A_{ik}-A_{kj}).
\]

\State Let $x$ and $\sigma$ be optimal column and row strategies
of $\mathcal S$, respectively, and let $a,b$ be the marginals
of $\sigma$.

\If{$v\le\eps$}
    \State \Return $(x,x)$
\EndIf

\If{$\max_{i\in[n]}e_i^\top A a-a^\top A a\le\eps$}
    \State \Return $(a,a)$
\EndIf

\State Define the near-best-response set
\[
    B\coloneqq\left\{i\in[n]:
    e_i^\top A a\ge
    \max_h e_h^\top A a-
    \frac{\alpha\eps}{1-\alpha}
    \right\}.
\]

\State Solve the restricted zero-sum game $\mathcal B$ with
row and column action sets $B$ and payoff matrix $A_{B,B}$.
Its value is
\[
    \rho =\min_{z\in\Delta(B)}
      \max_{i\in B}e_i^\top A z =\max_{w\in\Delta(B)}
      \min_{j\in B}w^\top A e_j.
\]

\If{$\rho\ge1-\eps$}
    \State Compute $w\in\Delta(B)$ satisfying
    \[
        w^\top A e_j\ge1-\eps
        \qquad\forall j\in B.
    \]
    \State \Return $(w,w)$
\Else
    \State Compute $z\in\Delta(B)$ satisfying
    \[
        e_i^\top A z\le1-\eps
        \qquad\forall i\in B.
    \]
    \State Set $p^+\coloneqq(1-\alpha)a+\alpha z$.
    \State \Return $(p^+,p^+)$
\EndIf

\end{algorithmic}
\end{algorithm}

\section{Warm-up: The Single-Repair Algorithm}\label{sec:warmup}

We begin with the single-repair algorithm, which uses the saddle-gap
certificate and at most one restricted zero-sum game. Its analysis
shows how the near-best-response set and the payoff guarantee from
$\mathcal S$ together yield an approximation factor below $1/3$.
The same estimates will be used later in \Cref{sec:amortized-algorithm,sec:amortized-analysis} for successive repairs in the
amortized algorithm.

Fix a repair weight $\alpha\in(0,1)$ and define
\[
\numberthis{eq:eps-alpha}
    \eps(\alpha)\coloneqq\frac{1-\alpha+\alpha^2}{3-2\alpha}.
\]

\begin{theorem}\label{thm:one-repair}
For every $\alpha\in(0,1)$, the single-repair algorithm
(\Cref{alg:warmup-one-repair}) computes a symmetric
$\eps(\alpha)$-approximate Nash equilibrium of $\mathcal G$.
Moreover,
\[
    \min_{\alpha\in(0,1)}\eps(\alpha)
    =\frac{\sqrt7}{2}-1,
\]
attained at $\alpha=(3-\sqrt7)/2$.
\end{theorem}

\begin{proof}
Fix $\alpha\in(0,1)$ and write $\eps\coloneqq\eps(\alpha)$.
We first verify the branches in which the algorithm returns a
strategy without performing a repair. If the saddle-gap game
$\mathcal S$ has value at most $\eps$, the output $(x,x)$ is
correct by \Cref{lem:low-saddle-gap}. Otherwise, the algorithm
checks the regret of the first marginal $a$ of an optimal row
strategy of $\mathcal S$. If this regret is at most $\eps$,
it returns $(a,a)$, which is correct by
\Cref{lem:symmetric-regret}.

The algorithm next forms the near-best-response set
\[
    B=\left\{i\in[n]:
    e_i^\top A a\ge\max_h e_h^\top A a
    -\frac{\alpha\eps}{1-\alpha}\right\}
\]
and solves the restricted zero-sum game $\mathcal B$ on
$B\times B$. If the value of $\mathcal B$ is at least
$1-\eps$, an optimal row strategy $w\in\Delta(B)$ satisfies
\[
    w^\top A e_j\ge1-\eps\qquad\forall j\in B.
\]
Since $w$ is supported on $B$,
\[
    w^\top A w
    =\sum_{j\in B}w_j w^\top A e_j
    \ge1-\eps.
\]
Every pure deviation in $\mathcal G$ earns at most one, so
$\operatorname{Reg}(w)\le\eps$. Hence the output $(w,w)$ is
also a symmetric $\eps$-approximate Nash equilibrium.

It remains to analyze the repair branch, in which the value of
$\mathcal B$ is below $1-\eps$. An optimal column strategy
$z\in\Delta(B)$ then satisfies
\[
    e_i^\top A z\le1-\eps\qquad\forall i\in B.
\]
The algorithm returns the symmetric profile $(p^+,p^+)$, where
\[
    p^+\coloneqq(1-\alpha)a+\alpha z.
\]
We first bound the best-response payoff against $p^+$.
For $i\in B$, the minimax guarantee for $z$ gives
\[
    e_i^\top A p^+
    &=(1-\alpha)e_i^\top A a+\alpha e_i^\top A z\notag\\
    &\le(1-\alpha)\max_h e_h^\top A a
        +\alpha(1-\eps)\notag\\
    &=\max_h e_h^\top A a
      -\alpha\left(\max_h e_h^\top A a+\eps-1\right).
      \numberthis{eq:warmup-inside}
\]
For $i\notin B$, the definition of $B$ gives
$e_i^\top A a<\max_h e_h^\top A a-\alpha\eps/(1-\alpha)$.
Although this row may earn as much as one against $z$, we have
\[
    e_i^\top A p^+
    &\le(1-\alpha)\left(\max_h e_h^\top A a
      -\frac{\alpha\eps}{1-\alpha}\right)+\alpha\notag\\
    &=\max_h e_h^\top A a
      -\alpha\left(\max_h e_h^\top A a+\eps-1\right).
      \numberthis{eq:warmup-outside}
\]
Thus, the restricted minimizer controls rows in $B$, while the
payoff gap in the definition of $B$ controls those outside it.
Both cases give the same upper bound, so
\[
\numberthis{eq:warmup-br}
    \max_i e_i^\top A p^+
    \le\max_h e_h^\top A a
    -\alpha\left(\max_h e_h^\top A a+\eps-1\right).
\]

We next lower-bound the self-payoff of $p^+$. Because $z$ is
supported on $B$, every row in its support is a near-best
response to $a$. Consequently,
\[
    z^\top A a
    =\sum_{i\in B}z_i e_i^\top A a
    \ge\max_h e_h^\top A a
      -\frac{\alpha\eps}{1-\alpha}.
\]
The high saddle-gap certificate also gives
$a^\top A e_j\ge\eps$ for every $j\in[n]$ by
\Cref{cor:initial-certificate}. Since $z\in\Delta(B)$,
this implies
\[
    a^\top A z\ge\eps.
\]
Expanding the self-payoff and using $z^\top A z\ge0$,
we obtain
\[
    (p^+)^\top A p^+
    &=(1-\alpha)^2a^\top A a
      +\alpha(1-\alpha)(a^\top A z+z^\top A a)
      +\alpha^2z^\top A z\notag\\
    &\ge(1-\alpha)^2a^\top A a
      +\alpha(1-\alpha)\left(\eps
        +\max_h e_h^\top A a
        -\frac{\alpha\eps}{1-\alpha}\right).
      \numberthis{eq:warmup-self}
\]
Subtracting \eqref{eq:warmup-self} from
\eqref{eq:warmup-br} gives
\[
    \operatorname{Reg}(p^+)
    &=\max_i e_i^\top A p^+-(p^+)^\top A p^+\notag\\
    &\le(1-\alpha)^2
      \left(\max_h e_h^\top A a-a^\top A a\right)\notag\\
    &\quad+\alpha(1-\eps)+\alpha^2\eps
      -\alpha(1-\alpha)\eps.
      \numberthis{eq:warmup-regret-repair}
\]
We are in the high saddle-gap branch, so
\Cref{cor:initial-certificate} gives $a^\top A a>2\eps$.
Since every pure payoff is at most one,
\[
    \max_h e_h^\top A a-a^\top A a\le1-2\eps.
\]
Substituting this bound into
\eqref{eq:warmup-regret-repair}, we get
\[
    \operatorname{Reg}(p^+)
    &\le(1-\alpha)^2(1-2\eps)
      +\alpha(1-\eps)+\alpha^2\eps
      -\alpha(1-\alpha)\eps\\
    &=1-\alpha+\alpha^2-2\eps(1-\alpha).
\]
By the definition of $\eps(\alpha)$ in
\eqref{eq:eps-alpha},
\[
    1-\alpha+\alpha^2-2\eps(\alpha)(1-\alpha)
    =\eps(\alpha).
\]
Hence $\operatorname{Reg}(p^+)\le\eps(\alpha)$, and
$(p^+,p^+)$ is a symmetric $\eps(\alpha)$-approximate
Nash equilibrium. This proves the approximation guarantee
in every branch of the algorithm.

We now optimize the repair weight. The bound just proved is
at most $\eps$ whenever
\[
    1-\alpha+\alpha^2-2\eps(1-\alpha)\le\eps,
\]
or equivalently,
\[
    \eps\ge\frac{1-\alpha+\alpha^2}{3-2\alpha}
    =\eps(\alpha).
\]
Differentiating gives
\[
    \eps'(\alpha)
    =\frac{-1+6\alpha-2\alpha^2}{(3-2\alpha)^2}.
\]
The unique critical point in $(0,1)$ is
$
    \alpha=(3-\sqrt7)/2$. The derivative is negative before this point and positive
after it, so the critical point minimizes $\eps(\alpha)$. Substitution gives
$
    \eps\left((3-\sqrt7)/2\right)
    =\frac{\sqrt7}{2}-1.
$
This completes the proof.
\end{proof}

The one-repair analysis uses the lower bound $a^\top A z\ge\eps$
to control the self-payoff of $p^+$. The update also retains
part of the initial guarantee against every pure column.
Indeed, for every $j\in[n]$,
\[
    (p^+)^\top A e_j
    &=(1-\alpha)a^\top A e_j+\alpha z^\top A e_j\notag\\
    &\ge(1-\alpha)a^\top A e_j
    \ge(1-\alpha)\eps,
\]
where the inequalities use nonnegativity of the payoffs and
\Cref{cor:initial-certificate}. The amortized algorithm
tracks this surviving guarantee and uses it in subsequent repairs.

\begin{algorithm}[!ht]
\caption{Saddle-Gap Algorithm with Amortized Repairs}
\label{alg:amortized-saddle-gap-repair}
\begin{algorithmic}[1]
\Require A symmetric game $\mathcal G=(A,A^\top)$ with
$A\in[0,1]^{n\times n}$ and integers $K\ge1$, $m\ge2$

\State Set $\eps\coloneqq\eps_{K,m}$.

\State Solve the saddle-gap game $\mathcal S$ with payoff matrix $H$
and value
\[
    v
    =\min_{x\in\Delta_n}
      \max_{(i,j)\in[n]^2}
      \sum_k x_k(A_{ik}-A_{kj})
    =\max_{\sigma\in\Delta_{n^2}}
      \min_{k\in[n]}
      \sum_{i,j}\sigma_{ij}(A_{ik}-A_{kj}).
\]

\State Let $x$ and $\sigma$ be optimal column and row strategies
of $\mathcal S$, respectively, and let $a,b$ be the marginals
of $\sigma$.

\If{$v\le\eps$}
    \State \Return $(x,x)$ \label{line:init}
\EndIf

\State Set $p_0\coloneqq a$.

\For{$t=0,1,\ldots,K-1$}

    \If{$\max_{i\in[n]}e_i^\top A p_t-p_t^\top A p_t\le\eps$}
        \State \Return $(p_t,p_t)$ \label{line:explicit_regret}
    \EndIf

    \State Define the near-best-response set
    \[
        B_t\coloneqq\left\{i\in[n]:
        e_i^\top A p_t\ge
        \max_h e_h^\top A p_t-
        \frac{\eps}{m+t-1}
        \right\}.
    \]

    \State Solve the restricted zero-sum game $\mathcal B_t$ with
    row and column action sets $B_t$ and payoff matrix
    $A_{B_t,B_t}$. Its value is
    \[
        \rho_t
        =\min_{z\in\Delta(B_t)}
          \max_{i\in B_t}e_i^\top A z
        =\max_{w\in\Delta(B_t)}
          \min_{j\in B_t}w^\top A e_j.
    \]

    \If{$\rho_t\ge1-\eps$}
        \State Compute $w_t\in\Delta(B_t)$ satisfying
        \label{line:w_cal}
        \[
            w_t^\top A e_j\ge1-\eps
            \qquad\forall j\in B_t.
        \]
        \State \Return $(w_t,w_t)$ \label{line:w_out}
    \Else
        \State Compute $z_t\in\Delta(B_t)$ satisfying
        \[
            e_i^\top A z_t\le1-\eps
            \qquad\forall i\in B_t.
        \]
        \State Set
        \[
            p_{t+1}\coloneqq
            \left(1-\frac1{m+t}\right)p_t
            +\frac1{m+t}z_t.
        \]
    \EndIf

\EndFor

\State \Return $(p_K,p_K)$

\end{algorithmic}
\end{algorithm}

\section{The Saddle-Gap Algorithm with Amortized Repairs}
\label{sec:amortized-algorithm}

The single-repair analysis shows that each update preserves part
of the payoff guarantee obtained from $\mathcal S$. We now use
this property to perform a sequence of repairs, carrying the
remaining guarantee from one iteration to the next.

Fix integers $K\ge1$ and $m\ge2$, and set
$\eps\coloneqq\eps_{K,m}$ as in \eqref{eq:epsilon-K-m}.
As in the single-repair algorithm, we first solve the saddle-gap
game $\mathcal S$. If its value $v$ is at most $\eps$, an optimal
column strategy $x$ yields a symmetric $\eps$-approximate Nash
equilibrium. Otherwise, we initialize $p_0\coloneqq a$, where
$a$ is the first marginal of an optimal row strategy of
$\mathcal S$. By \Cref{cor:initial-certificate}, this
initial strategy satisfies
$\operatorname{Reg}(p_0)\le1-2\eps$ and
$p_0^\top A e_j\ge\eps$ for every $j\in[n]$.

The algorithm performs at most $K$ repairs, using the weights
\[
\numberthis{eq:weights}
    \alpha_t\coloneqq\frac1{m+t},
    \qquad t=0,1,\ldots,K-1.
\]
The corresponding near-best-response window at iteration $t$
has width
\[
\numberthis{eq:window}
    \frac{\alpha_t\eps}{1-\alpha_t}
    =\frac{\eps}{m+t-1}.
\]

At each iteration, the algorithm first checks whether
$\operatorname{Reg}(p_t)\le\eps$, in which case it returns
$(p_t,p_t)$. Otherwise, it forms the near-best-response set
$B_t$ using the window in \eqref{eq:window} and solves the
restricted zero-sum game $\mathcal B_t$ with payoff matrix
$A_{B_t,B_t}$ and value $\rho_t$.

If $\rho_t\ge1-\eps$, an optimal row strategy
$w_t\in\Delta(B_t)$ guarantees payoff at least $1-\eps$
against every pure column in $B_t$. Since $w_t$ is supported
on $B_t$, the algorithm can return the symmetric
$\eps$-approximate Nash equilibrium $(w_t,w_t)$.
If $\rho_t<1-\eps$, an optimal column strategy
$z_t\in\Delta(B_t)$ guarantees
$e_i^\top A z_t\le1-\eps$ for every $i\in B_t$.
The algorithm uses $z_t$ as a repair direction and updates
\[
    p_{t+1}\coloneqq(1-\alpha_t)p_t+\alpha_t z_t.
\]
If all $K$ repairs are completed, the algorithm returns
$(p_K,p_K)$.

The decreasing repair weights preserve a portion of the
column-payoff guarantee at each step. The next section shows
how this guarantee controls the regret across successive
repairs and establishes the approximation bound
$\eps_{K,m}$.

The complete procedure is given in
\Cref{alg:amortized-saddle-gap-repair}. It solves $\mathcal S$,
with $n^2$ row actions and $n$ column actions, followed by
at most $K$ restricted zero-sum games $\mathcal B_t$,
each with at most $n$ actions per player.

\section{Analysis of the Amortized Algorithm}\label{sec:amortized-analysis}

We begin by verifying the early-return branches of \Cref{alg:amortized-saddle-gap-repair} and proving a
one-step repair inequality. We then apply this inequality across
successive iterations, solve the resulting recurrences, and use
these bounds to establish the main theorem and its corollaries.

\subsection{Early Termination and One-Step Repair}

The low-value branch of $\mathcal S$ in
Line~\ref{line:init} of
\Cref{alg:amortized-saddle-gap-repair} is correct by
\Cref{lem:low-saddle-gap}. Similarly, the explicit regret
check in Line~\ref{line:explicit_regret} is correct by
\Cref{lem:symmetric-regret}. We therefore consider the
remaining early-return branch, in which a restricted zero-sum
game has value at least $1-\eps$.

\begin{lemma}\label{lem:restricted-high-branch}
Suppose \Cref{alg:amortized-saddle-gap-repair} reaches
Line~\ref{line:w_out} at iteration $t$, with
$\rho_t\ge1-\eps$. Then the strategy $w_t$ computed in
Line~\ref{line:w_cal} yields a symmetric
$\eps$-approximate Nash equilibrium $(w_t,w_t)$ of
$\mathcal G$.
\end{lemma}

\begin{proof}
By \Cref{lem:minimax-operational}, the strategy
$w_t\in\Delta(B_t)$ satisfies
\[
    w_t^\top A e_j\ge1-\eps
    \qquad\forall j\in B_t.
\]
Since $w_t$ is supported on $B_t$, its self-payoff satisfies
\[
    w_t^\top A w_t
    =\sum_{j\in B_t}(w_t)_j w_t^\top A e_j
    \ge1-\eps.
\]
Every pure deviation in $\mathcal G$ has payoff at most one.
Therefore,
\[
    \max_{i\in[n]}e_i^\top A w_t-w_t^\top A w_t
    \le1-(1-\eps)=\eps.
\]
The claim follows from \Cref{lem:symmetric-regret}.
\end{proof}

We now analyze the repair branch. For a current strategy
$p$, we use $R$ to upper-bound its regret and $L$ to
lower-bound its payoff against every pure column strategy.
The following lemma establishes how these two bounds
change after one repair.

\begin{lemma}\label{lem:one-step-repair}
Fix $\eps\in[0,1]$ and $\alpha\in(0,1)$. Let
$p\in\Delta_n$ satisfy
\[
\numberthis{eq:repair-regret-input}
    \max_{i\in[n]}e_i^\top A p-p^\top A p\le R
\]
for some $R\ge0$, and
\[
    p^\top A e_j\ge L\qquad\forall j\in[n]
\]
for some $L\ge0$. Define the near-best-response set
\[
    B\coloneqq\left\{i\in[n]:
    e_i^\top A p\ge
    \max_h e_h^\top A p-
    \frac{\alpha\eps}{1-\alpha}
    \right\}.
\]
Suppose the restricted zero-sum game $\mathcal B$ with
row and column action sets $B$ and payoff matrix $A_{B,B}$
has value less than $1-\eps$. Let $z\in\Delta(B)$ be an
optimal column strategy of $\mathcal B$, so that
\[
    e_i^\top A z\le1-\eps
    \qquad\forall i\in B.
\]
Then the repaired strategy
$p^+\coloneqq(1-\alpha)p+\alpha z$ satisfies
\[
\numberthis{eq:repair-floor-output}
    (p^+)^\top A e_j\ge(1-\alpha)L
    \qquad\forall j\in[n],
\]
and
\[
\numberthis{eq:repair-regret-output}
    \max_{i\in[n]}e_i^\top A p^+-(p^+)^\top A p^+\le(1-\alpha)^2R+\alpha(1-\eps)+\alpha^2\eps-\alpha(1-\alpha)L.
\]
\end{lemma}

\begin{proof}
We first establish the column-payoff bound of \eqref{eq:repair-floor-output}. For every
$j\in[n]$, nonnegativity of the payoffs gives
\[
    (p^+)^\top A e_j
    =(1-\alpha)p^\top A e_j+\alpha z^\top A e_j \ge(1-\alpha)L.
\]
We next upper-bound the best-response payoff against
$p^+$ in \eqref{eq:repair-regret-output}. Consider a row $i\in B$. The restricted minimizer
satisfies $e_i^\top A z\le1-\eps$, so
\[
    e_i^\top A p^+
    &=(1-\alpha)e_i^\top A p+\alpha e_i^\top A z\notag\\
    &\le(1-\alpha)\max_h e_h^\top A p
       +\alpha(1-\eps)\notag\\
    &=\max_h e_h^\top A p
      -\alpha\left(\max_h e_h^\top A p+\eps-1\right).
      \numberthis{eq:repair-inside}
\]
Now consider a row $i\notin B$. By the definition of $B$,
\[
    e_i^\top A p
    <\max_h e_h^\top A p-
      \frac{\alpha\eps}{1-\alpha}.
\]
Although this row may obtain payoff as large as one
against $z$, we have
\[
    e_i^\top A p^+
    &=(1-\alpha)e_i^\top A p+\alpha e_i^\top A z\notag\\
    &\le(1-\alpha)\left(\max_h e_h^\top A p
      -\frac{\alpha\eps}{1-\alpha}\right)+\alpha\notag\\
    &=\max_h e_h^\top A p
      -\alpha\left(\max_h e_h^\top A p+\eps-1\right).
      \numberthis{eq:repair-outside}
\]
The two cases yield the same upper bound. For rows in
$B$, this bound follows from the minimax guarantee for
$z$. For rows outside $B$, it follows from their payoff
gap against $p$. This explains the choice of the
near-best-response window $\alpha\eps/(1-\alpha)$.
Combining \eqref{eq:repair-inside} and
\eqref{eq:repair-outside}, we obtain
\[
\numberthis{eq:repair-br-bound}
    \max_{i\in[n]}e_i^\top A p^+
    \le\max_h e_h^\top A p
       -\alpha\left(\max_h e_h^\top A p+\eps-1\right).
\]
We now lower-bound the self-payoff of $p^+$.
Since $z$ is supported on $B$, every row in its
support is a near-best response to $p$. Therefore,
\[
    z^\top A p
    &=\sum_{i\in B}z_i e_i^\top A p\\
    &\ge\max_h e_h^\top A p-
       \frac{\alpha\eps}{1-\alpha}.
\]
The column-payoff bound on $p$ gives a lower bound
for the other cross term. Since $z$ is a probability
distribution and $p^\top A e_j\ge L$ for every $j$,
\[
    p^\top A z
    =\sum_{j\in B}z_j p^\top A e_j
    \ge L.
\]

Expanding the self-payoff of $p^+$, we have
\[
    (p^+)^\top A p^+
    &=(1-\alpha)^2p^\top A p
      +\alpha(1-\alpha)(p^\top A z+z^\top A p)
      +\alpha^2z^\top A z\notag\\
    &\ge(1-\alpha)^2p^\top A p
      +\alpha(1-\alpha)\left(
        L+\max_h e_h^\top A p
        -\frac{\alpha\eps}{1-\alpha}
      \right).
      \numberthis{eq:repair-self-bound}
\]
Here we used $z^\top A z\ge0$, which follows from the
nonnegativity of the entries of $A$.

Subtracting \eqref{eq:repair-self-bound} from
\eqref{eq:repair-br-bound}, we obtain
\[
    \max_{i\in[n]}e_i^\top A p^+-(p^+)^\top A p^+
    &\le(1-\alpha)^2
      \left(\max_h e_h^\top A p-p^\top A p\right)\notag\\
    &\quad+\alpha(1-\eps)+\alpha^2\eps
      -\alpha(1-\alpha)L\notag\\
    &\le(1-\alpha)^2R+\alpha(1-\eps)+\alpha^2\eps
      -\alpha(1-\alpha)L,
\]
where the last inequality follows from
\eqref{eq:repair-regret-input}. This proves
\eqref{eq:repair-regret-output}.
\end{proof}

The regret bound in \eqref{eq:repair-regret-output}
retains a factor $(1-\alpha)^2$ of the previous regret
budget. The terms $\alpha(1-\eps)$ and $\alpha^2\eps$
arise from the best-response bound and the
near-best-response window, while
$-\alpha(1-\alpha)L$ captures the contribution of
the column-payoff floor to the new self-payoff.
By \eqref{eq:repair-floor-output}, the repaired strategy
also retains a column-payoff floor of at least
$(1-\alpha)L$. The next subsection uses these two
bounds across successive repairs.

\subsection{Amortized Regret Analysis}

We now apply \Cref{lem:one-step-repair} to the sequence of
repairs in \Cref{alg:amortized-saddle-gap-repair}. Suppose the
saddle-gap game has value $v>\eps$, so the algorithm
initializes $p_0=a$. By \Cref{cor:initial-certificate},
\[
    \max_{i\in[n]}e_i^\top A p_0-p_0^\top A p_0\le1-2\eps,
    \qquad p_0^\top A e_j\ge\eps\quad\forall j\in[n].
\]

We track the regret of $p_t$ through an upper bound $R_t$
and its payoff against every pure column through a lower
bound $L_t$. These are quantities used in the analysis.
We initialize
\[
\numberthis{eq:budgets-initial}
    R_0\coloneqq1-2\eps,\qquad L_0\coloneqq\eps,
\]
and, for $t=0,1,\ldots,K-1$, define
\[
    R_{t+1}&\coloneqq(1-\alpha_t)^2R_t+\alpha_t(1-\eps)+\alpha_t^2\eps-\alpha_t(1-\alpha_t)L_t,
    \numberthis{eq:budget-recurrence_R}\\
    L_{t+1}&\coloneqq(1-\alpha_t)L_t.
    \numberthis{eq:budget-recurrence_L}
\]
Here $\alpha_t=1/(m+t)$, as defined in \eqref{eq:weights}.
These recurrences follow from the one-step repair bounds.
The next lemma verifies that they remain valid throughout
the algorithm.

\begin{lemma}\label{lem:invariant}
For every $t=0,1,\ldots,K$, if the algorithm reaches
the strategy $p_t$, then
\[
    \max_{i\in[n]}e_i^\top A p_t-p_t^\top A p_t\le R_t,
    \qquad p_t^\top A e_j\ge L_t\quad\forall j\in[n].
\]
\end{lemma}

\begin{proof}
At $t=0$, the claim follows from
\eqref{eq:budgets-initial} and
\Cref{cor:initial-certificate}.

Suppose the claim holds at some $t<K$ and the algorithm
performs a repair. The repair weight is
$\alpha_t=1/(m+t)$, and by \eqref{eq:window}, the
near-best-response window satisfies
\[
    \frac{\alpha_t\eps}{1-\alpha_t}
    =\frac{\eps}{m+t-1}.
\]
Thus, \Cref{lem:one-step-repair} applies to the current
strategy $p_t$ with regret budget $R_t$ and column-floor
budget $L_t$. Its two conclusions give
\[
    \max_{i\in[n]}e_i^\top A p_{t+1}-p_{t+1}^\top A p_{t+1}
    \le R_{t+1},
    \qquad p_{t+1}^\top A e_j\ge L_{t+1}\quad\forall j\in[n],
\]
where $R_{t+1}$ and $L_{t+1}$ are exactly the quantities
defined in \eqref{eq:budget-recurrence_R} and
\eqref{eq:budget-recurrence_L}. The claim follows by
induction.
\end{proof}

We next solve the budget recurrences for $L_t$ and $R_t$. The floor recurrence
can be solved independently of the regret recurrence.
Using $L_0=\eps$ and $\alpha_t=1/(m+t)$, we obtain
\[
\numberthis{eq:floor-closed}
    L_t=\eps\prod_{r=0}^{t-1}\frac{m+r-1}{m+r}
    =\eps\frac{m-1}{m+t-1}.
\]
In particular, $(m+t-1)L_t=\eps(m-1)$ is independent
of $t$. This identity allows the regret recurrence to
telescope after scaling.

\begin{lemma}\label{lem:closed-form}
For the budgets defined in \eqref{eq:budgets-initial},
\eqref{eq:budget-recurrence_R}, and
\eqref{eq:budget-recurrence_L}, the choice
$\eps=\eps_{K,m}$ from \eqref{eq:epsilon-K-m} satisfies $
    R_K=\eps$.
\end{lemma}

\begin{proof}
Substituting $\alpha_t=1/(m+t)$ into
\eqref{eq:budget-recurrence_R} and multiplying both sides
by $(m+t)^2$, we obtain
\[
    (m+t)^2R_{t+1}
    &=(m+t)^2\left(1-\frac1{m+t}\right)^2R_t
      +(m+t)(1-\eps)+\eps -(m+t)^2\frac1{m+t}
      \left(1-\frac1{m+t}\right)L_t\notag\\
    &=(m+t-1)^2R_t+(m+t)(1-\eps)+\eps-(m+t-1)L_t.
\]
By \eqref{eq:floor-closed},
$(m+t-1)L_t=\eps(m-1)$. Substituting this identity
gives
\[
    (m+t)^2R_{t+1}
    &=(m+t-1)^2R_t+(m+t)(1-\eps)+\eps-\eps(m-1)\notag\\
    &=(m+t-1)^2R_t+(m+t)-\eps(2m+t-2).
    \numberthis{eq:telescoping}
\]

To make the telescoping explicit, define
\[
    S_t\coloneqq(m+t-1)^2R_t.
\]
Equation~\eqref{eq:telescoping} becomes
\[
    S_{t+1}=S_t+(m+t)-\eps(2m+t-2).
\]
Summing over $t=0,1,\ldots,K-1$, we obtain
\[
    S_K=S_0+\sum_{t=0}^{K-1}(m+t)
        -\eps\sum_{t=0}^{K-1}(2m+t-2).
\]
The two sums satisfy
\[
    \sum_{t=0}^{K-1}(m+t)&=Km+\frac{K(K-1)}2,\\
    \sum_{t=0}^{K-1}(2m+t-2)&=K(2m-2)+\frac{K(K-1)}2.
\]
Since $S_0=(m-1)^2R_0=(m-1)^2(1-2\eps)$ and
$S_K=(m+K-1)^2R_K$, we conclude that
\[
    (m+K-1)^2R_K
    &=(m-1)^2(1-2\eps)+Km+\frac{K(K-1)}2 -\eps\left(K(2m-2)+\frac{K(K-1)}2\right).
    \numberthis{eq:RK-expression}
\]

We determine the approximation factor by requiring the
final regret budget to equal the target guarantee,
$R_K=\eps$. Substituting this equality into
\eqref{eq:RK-expression} gives
\[
    (m+K-1)^2\eps
    &=(m-1)^2+Km+\frac{K(K-1)}2 -\eps\left(2(m-1)^2+K(2m-2)
      +\frac{K(K-1)}2\right).
\]
Collecting the terms involving $\eps$ and solving yields
\[
    \eps=
    \frac{(m-1)^2+Km+K(K-1)/2}
    {(m+K-1)^2+2(m-1)^2+K(2m-2)+K(K-1)/2}.
\]
Multiplying numerator and denominator by two and
simplifying, we obtain
\[
    \eps=
    \frac{2m^2+2(K-2)m+K(K-1)+2}
    {6m^2+(8K-12)m+3(K-1)(K-2)}
    =\eps_{K,m}.
\]
Thus, the choice of $\eps$ in \eqref{eq:epsilon-K-m}
indeed gives $R_K=\eps$, as claimed.
\end{proof}

\subsection{Proofs of the Main Results}

We now combine the preceding lemmas to prove the main theorem
and derive its two corollaries.

\begin{repeatthm}{thm:finite-main}[Main result]
For every pair of integers $K\ge1$ and $m\ge2$, there is a
polynomial-time algorithm that, given a symmetric bimatrix game
$\mathcal G=(A,A^\top)$ with $A\in[0,1]^{n\times n}$, computes
a symmetric $\eps_{K,m}$-approximate Nash equilibrium, where
\[
\eps_{K,m}=
\frac{2m^2+2(K-2)m+K(K-1)+2}
{6m^2+(8K-12)m+3(K-1)(K-2)} .
\]
The algorithm solves the auxiliary zero-sum game $\mathcal S$,
with $n^2$ row actions and $n$ column actions, followed by at
most $K$ restricted zero-sum games $\mathcal B_t$, each with
at most $n$ actions per player.
\end{repeatthm}

\begin{proof}
Write $\eps=\eps_{K,m}$. The algorithm can terminate in
one of three early-return branches or after completing
all $K$ repairs.

If it returns in Line~\ref{line:init} of
\Cref{alg:amortized-saddle-gap-repair}, the output is a
symmetric $\eps$-approximate Nash equilibrium by
\Cref{lem:low-saddle-gap}. The regret check in
Line~\ref{line:explicit_regret} is correct by
\Cref{lem:symmetric-regret}, and the restricted high-value
branch in Line~\ref{line:w_out} is correct by
\Cref{lem:restricted-high-branch}.

It remains to consider the case in which the algorithm
performs all $K$ repairs and returns $(p_K,p_K)$.
By \Cref{lem:invariant},
\[
    \max_{i\in[n]}e_i^\top A p_K-p_K^\top A p_K\le R_K.
\]
Since \Cref{lem:closed-form} establishes that
$R_K=\eps_{K,m}=\eps$, we have
\[
    \max_{i\in[n]}e_i^\top A p_K-p_K^\top A p_K\le\eps.
\]
Thus, \Cref{lem:symmetric-regret} implies that
$(p_K,p_K)$ is a symmetric $\eps$-approximate Nash
equilibrium.

Finally, the algorithm solves one zero-sum game
$\mathcal S$ with $n^2$ row actions and $n$ column actions,
followed by at most $K$ restricted zero-sum games
$\mathcal B_t$, each with at most $n$ actions per player.
All these games can be solved in polynomial time by
linear programming, and the remaining operations are
also polynomial-time computable. This establishes the
claimed running time.
\end{proof}

\begin{repeatcorollary}{cor:three-tenths}
Every symmetric bimatrix game with payoffs in $[0,1]$ has a
polynomial-time computable symmetric $3/10$-approximate Nash equilibrium.
\end{repeatcorollary}

\begin{proof}
Setting $K=13$ and $m=14$ in \eqref{eq:epsilon-K-m} gives
\[
    \eps_{13,14}
    =\frac{2\cdot14^2+2\cdot11\cdot14+13\cdot12+2}
          {6\cdot14^2+(8\cdot13-12)\cdot14+3\cdot12\cdot11}
    =\frac{858}{2860}=\frac3{10}.
\]
The claim follows from \Cref{thm:finite-main}.
\end{proof}

The finite guarantee in \eqref{eq:epsilon-K-m} also yields
an improved asymptotic approximation factor when $m$ grows
proportionally to $K$. We optimize this limiting factor
to obtain the following result.

\begin{repeatcorollary}{cor:limit}
For every fixed $\delta>0$, there is a polynomial-time algorithm
that computes a symmetric
$\left(1-\frac1{\sqrt2}+\delta\right)$-approximate Nash equilibrium
of every symmetric bimatrix game with payoffs in $[0,1]$.
\end{repeatcorollary}

\begin{proof}
Consider a sequence of integers $K,m$ such that
$K\to\infty$ and $m/K\to c$ for some constant $c>0$.
Dividing the numerator and denominator of
\eqref{eq:epsilon-K-m} by $K^2$ gives
\[
    \eps_{K,m}\longrightarrow
    g(c)\coloneqq\frac{2c^2+2c+1}{6c^2+8c+3}.
\]
We minimize $g$ over $c>0$. Differentiating yields
\[
    g'(c)=\frac{4c^2-2}{(6c^2+8c+3)^2}.
\]
Since the denominator is positive for $c>0$, the
derivative vanishes precisely when $4c^2-2=0$.
The unique positive critical point is $
    c=1/{\sqrt2}.$
Moreover, $g'(c)<0$ for $0<c<1/\sqrt2$ and
$g'(c)>0$ for $c>1/\sqrt2$. Thus, this critical
point is the global minimizer of $g$ over $c>0$.
Evaluating $g$ at this point gives
\[
    g\left(\frac1{\sqrt2}\right)
    =\frac{2+\sqrt2}{6+4\sqrt2}
    =1-\frac1{\sqrt2}.
\]
For each $K$, choose
\[
    m\coloneqq\left\lceil\frac{K}{\sqrt2}\right\rceil.
\]
Then $m/K\to1/\sqrt2$, and consequently
\[
    \eps_{K,m}\longrightarrow1-\frac1{\sqrt2}
    \qquad\text{as }K\to\infty.
\]
Hence, for every fixed $\delta>0$, we can choose
$K$ sufficiently large, with $m\ge2$, such that
\[
    \eps_{K,m}\le1-\frac1{\sqrt2}+\delta.
\]
For a fixed $\delta$, these choices of $K$ and $m$
are constants independent of the input size.
By \Cref{thm:finite-main}, the resulting algorithm
runs in polynomial time and computes a symmetric
$\left(1-\frac1{\sqrt2}+\delta\right)$-approximate
Nash equilibrium.
\end{proof}

\section{Conclusion}\label{sec:conclusion}

We have shown that the $1/3$ approximation barrier for
symmetric bimatrix games can be broken in polynomial time
using auxiliary zero-sum games as the only optimization
subroutines. For every fixed $\delta>0$, our algorithm
computes a symmetric
$\left(1-1/\sqrt2+\delta\right)$-approximate Nash equilibrium
for every symmetric bimatrix game. This improves upon the
best currently known polynomial-time approximation
guarantee for general bimatrix games.

Our analysis achieves the limiting factor $1-1/\sqrt2$
for the repair schedule $\alpha_t=1/(m+t)$. Different
weight schedules or more refined repair rules may lead
to further improvements. One possible direction is to
exploit the full joint distribution over ordered pairs
in the saddle-gap solution, beyond the marginal information
used by our algorithm. Understanding the hardness of
computing approximate symmetric Nash equilibria and
determining the optimal polynomial-time approximation
threshold for this class also remain important open
problems. It would also be interesting to investigate
whether similar techniques can improve the $1/3$
approximation frontier for general bimatrix games.
More broadly, the use of auxiliary zero-sum games to
guide successive local repairs may prove useful in
other approximation problems.

\section*{Acknowledgement}

A.S. and P.J. were partially supported by the ONR grant N00014-24-1-2470. Part of this research was performed while A.S. was visiting the Institute for Mathematical and Statistical Innovation (IMSI), which is supported by the National Science Foundation (Grant No. DMS-2425650). G.F. was supported in part by the National Science Foundation awards CCF-2443068 and IIS-2552046, the Office of Naval Research grant N000142512296, and an AI2050 Early Career Fellowship.

\bibliographystyle{unsrtnat}
\bibliography{refs}

\end{document}

%% file: notation.tex
\def\[#1\]{\begin{align*}#1\end{align*}}

\NewDocumentCommand{\numberthis}{om}{%
  \IfNoValueTF{#1}{%
    \refstepcounter{equation}\tag{\theequation}%
  }{%
    \tag{#1}%
  }%
  \label{#2}%
}

\newtheorem{theorem}{Theorem}[section]

\newaliascnt{lemma}{theorem}
\newtheorem{lemma}[lemma]{Lemma}
\aliascntresetthe{lemma}

\newaliascnt{proposition}{theorem}

\aliascntresetthe{proposition}

\newaliascnt{corollary}{theorem}
\newtheorem{corollary}[corollary]{Corollary}
\aliascntresetthe{corollary}

\newaliascnt{claim}{theorem}

\aliascntresetthe{claim}

\theoremstyle{definition}
\newaliascnt{definition}{theorem}

\aliascntresetthe{definition}

\newaliascnt{question}{theorem}

\aliascntresetthe{question}

\theoremstyle{remark}
\newaliascnt{remark}{theorem}

\aliascntresetthe{remark}

\crefname{theorem}{theorem}{theorems}
\Crefname{theorem}{Theorem}{Theorems}
\crefname{lemma}{lemma}{lemmas}
\Crefname{lemma}{Lemma}{Lemmas}
\crefname{proposition}{proposition}{propositions}
\Crefname{proposition}{Proposition}{Propositions}
\crefname{corollary}{corollary}{corollaries}
\Crefname{corollary}{Corollary}{Corollaries}
\crefname{definition}{definition}{definitions}
\Crefname{definition}{Definition}{Definitions}
\crefname{section}{section}{sections}
\Crefname{section}{Section}{Sections}

\newcommand{\eps}{\varepsilon}

\numberwithin{equation}{section}

\theoremstyle{plain}

\newcommand{\repeatthmname}{}
\newtheorem*{repeatthminner}{\repeatthmname}
\newenvironment{repeatthm}[1]{%
  \renewcommand{\repeatthmname}{Theorem~\ref{#1}}%
  \begin{repeatthminner}%
}{%
  \end{repeatthminner}%
}

\newcommand{\repeatlemmaname}{}
\newtheorem*{repeatlemmainner}{\repeatlemmaname}

\newcommand{\repeatpropname}{}
\newtheorem*{repeatpropinner}{\repeatpropname}

\newcommand{\repeatcorollaryname}{}
\newtheorem*{repeatcorollaryinner}{\repeatcorollaryname}
\newenvironment{repeatcorollary}[1]{%
  \renewcommand{\repeatcorollaryname}{Corollary~\ref{#1}}%
  \begin{repeatcorollaryinner}%
}{%
  \end{repeatcorollaryinner}%
}